\PassOptionsToPackage{dvipsnames,table}{xcolor}
\documentclass[final,twoside,11pt]{entics}
\usepackage{enticsmacro}
\usepackage[fixamsmath]{mathtools}
\usepackage{mathcommand}
\usepackage{graphicx}
\usepackage{cmll}
\usepackage{mathpartir}
\usepackage[all]{xy}
\usepackage[inline]{enumitem}
\usepackage{fdsymbol}
\usepackage{stmaryrd}
\usepackage{parskip}
\usepackage[inline]{enumitem}
\usepackage{xspace}
\usepackage{multirow}
\usepackage{longtable}
\usepackage[UKenglish,all]{foreign}
\usepackage[goodsyntax]{virginialake}
\vllineartrue%
\def\conf{MFPS 2024}
\volume{4}%
\def\volu{NN}%
\def\papno{nn}%

\def\lastname{Atkey, Kokke}
\def\runtitle{Semantic Cut Elimation for MAV}
\def\copynames{R. Atkey, W. Kokke}

\def\C{Ch.\ }
\def\Fn{Fn.\ }
\renewcommand{\S}{Sec.\ }

\usepackage{iftex}
\ifPDFTeX
  \usepackage[T1,T5]{fontenc}
  \usepackage[utf8]{inputenc}
  \usepackage[english]{babel}
  \usepackage{newunicodechar}
\else 
  \usepackage{unicode-math}
  \usepackage{newunicodechar}
\fi
\newunicodechar{∼}{\ensuremath{\mathnormal\sim}}
\newunicodechar{≃}{\ensuremath{\mathnormal\simeq}}
\newunicodechar{⟶}{\ensuremath{\mathnormal\longrightarrow}}
\newunicodechar{⋆}{\ensuremath{\mathnormal\star}}
\newunicodechar{⟷}{\ensuremath{\mathnormal\longleftrightarrow}}
\newunicodechar{⟦}{\ensuremath{\mathnormal\llbracket}}
\newunicodechar{⟧}{\ensuremath{\mathnormal\rrbracket}}

\ExplSyntaxOn
\cs_new_eq:NN \IfSingleTF \tl_if_single:nTF
\ExplSyntaxOff

\usepackage[nameinlink]{cleveref}
\makeatletter
\newcommand{\customlabel}[4][0]{%
	\protected@write\@auxout{}{\string\newlabel{#3}{{#4}{\thepage}{#4}{#3}{}}}%
	\protected@write\@auxout{}{\string\newlabel{#3@cref}{{[#2][#1][#1]#4}{\thepage}}}}
\makeatother

\crefname{rule}{rule}{rules}
\crefname{defn}{definition}{definitions}
\crefname{lem}{lemma}{lemmas}
\crefname{thm}{theorem}{theorems}
\crefname{prop}{proposition}{propositions}

\NewDocumentCommand{\RuleName}{m}{\textsc{#1}}
\NewDocumentCommand{\RuleLabel}{s o m}{%
  \IfValueTF{#2}{%
    \customlabel{Rule}{#2}{\RuleName{#3}}%
    \IfBooleanF{#1}{\hypertarget{#2}{\RuleName{#3}}}
  }{%
    \customlabel{Rule}{rule:#3}{\RuleName{#3}}%
    \IfBooleanF{#1}{\hypertarget{rule:#3}{\RuleName{#3}}}
  }}
\crefformat{Rule}{#2(#1)#3}

\makeatletter
\newcommand*{\inlineequation}[2][]{%
  \begingroup
    \refstepcounter{equation}%
    \ifx\\#1\\%
    \else
      \label{#1}%
    \fi
    \relpenalty=10000 %
    \binoppenalty=10000 %
    \ensuremath{%
      #2%
    }%
    ~\@eqnnum%
  \endgroup
}
\makeatother

\def\AgdaBaseUrl{https://bobatkey.github.io/semantic-cut-elimination/MFPS/2024}
\NewDocumentCommand{\AgdaModuleRef}{m}{%
  \href{%
    \AgdaBaseUrl/#1.html}{%
    \texttt{#1}}}
\NewDocumentCommand{\AgdaRef}{s m v o}{%
  \href{%
    \AgdaBaseUrl/#2.html\##3}{%
    \texttt{%
      \IfBooleanF{#1}{#2.}%
      \IfValueTF{#4}{#4}{#3}}}}

\newcommand{\longtableheader}[1]{%
  \rowcolor{lightgray}
  \multicolumn{3}{l}{#1}}
\newcommand{\longtablemodule}[1]{%
  \multicolumn{3}{r}{%
    \small%
    (The preceding definitions can be found under \AgdaModuleRef{#1}.)}}

\newcommandPIE\va[0]{\alpha#1#2#3}
\newcommandPIE\vb[0]{\beta#1#2#3}
\newcommandPIE\vP[0]{P#1#2#3}
\newcommandPIE\vQ[0]{Q#1#2#3}
\newcommandPIE\vR[0]{R#1#2#3}
\newcommandPIE\vS[0]{S#1#2#3}
\newcommandPIE\vC[0]{\mathcal{C}#1#2#3}
\newcommandPIE\vD[0]{\mathcal{D}#1#2#3}
\newcommandPIE\vGG[0]{\Gamma#1#2#3}
\newcommandPIE\vGD[0]{\Delta#1#2#3}
\newcommandPIE\vSS[0]{\mathcal{S}#1#2#3}

\newcommand{\vDual}[1]{\overline{#1}}
\newcommand{\vPos}[1]{#1}
\newcommand{\vNeg}[1]{\vDual{#1}}

\newlength{\vPWidth}
\def\vUnit{%
  \ensuremath{%
    \mathord{%
      \makebox[\vPWidth][c]{\ensuremath{\mathnormal{I}}}}}}

\usepackage{tikz}
\newlength{\vlpaWidth}
\newlength{\vlseWidth}
\newlength{\vOpWidth}
\pgfmathsetlength{\vOpWidth}{max(\vlpaWidth,\vlseWidth)}

\def\vTens{%
  \ensuremath{%
    \mathbin{%
      \makebox[\vOpWidth][c]{\ensuremath{\vlte}}}}}
\def\vParr{%
  \ensuremath{%
    \mathbin{%
      \makebox[\vOpWidth][c]{\ensuremath{\vlpa}}}}}
\def\vPlus{%
  \ensuremath{%
    \mathbin{%
      \makebox[\vOpWidth][c]{\ensuremath{\vlor}}}}}
\def\vWith{%
  \ensuremath{%
    \mathbin{%
      \makebox[\vOpWidth][c]{\ensuremath{\vlan}}}}}
\def\vSeq{%
  \ensuremath{%
    \mathbin{%
      \makebox[\vOpWidth][c]{\ensuremath{\vlse}}}}}

\newcommand{\vEmpty}[0]{\ensuremath{\vlhole}}
\newcommand{\vPlug}[1]{\ensuremath{\{#1\}}}

\NewDocumentCommand{\vEquiv}{s}{%
  \ensuremath{\simeq}}
\NewDocumentCommand{\vInferFrom}{s}{%
  \ensuremath{%
    \mathrel{%
      \longrightarrow\IfBooleanTF{#1}{^\star}{}}}}
\NewDocumentCommand{\vInferTo}{s}{%
  \ensuremath{%
    \mathrel{%
      \longleftarrow\IfBooleanTF{#1}{^\star}{}}}}
\NewDocumentCommand{\vInferFromTo}{s}{%
  \ensuremath{%
    \mathrel{%
      \longleftrightarrow\IfBooleanTF{#1}{^\star}{}}}}

\newcommand{\LowerSet}[1]{\widehat{#1}}
\newcommand{\Day}[1]{\mathop{\widehat{#1}}}
\newcommand{\ClosedLowerSet}[1]{\widehat{#1}^+}
\newcommand{\ClosedDay}[1]{\mathop{\widehat{#1}^+}}
\newcommand{\Chu}{\mathrm{Chu}}
\newcommand{\op}{\mathrm{op}}
\newcommand{\NMAV}{\textsc{NMav}\xspace}
\newcommand{\sem}[1]{\llbracket #1 \rrbracket}

\newcommand{\aDual}[1]{\neg\IfSingleTF{#1}{#1}{(#1)}}

\newcommand{\ChuEmbed}{\eta^c}
\newcommand{\ClosedLowerEmbed}{\eta^+}
\newcommand{\LowerEmbed}{\eta}

\usepackage{todonotes}
\presetkeys{todonotes}{inline}{}

\def\conf{MFPS 2024}
\def\myconf{mfps40} 	
\volume{4}			
\def\volu{4}			
\def\papno{4}			
\def\mydoi{14870}%

\begin{document}

\begin{frontmatter}
  \title{A Semantic Proof of\\[1ex] Generalised Cut Elimination for Deep Inference\thanksref{ALL}}
  \thanks[ALL]{This work
    was funded by the \href{https://www.ukri.org/about-us/epsrc/}{Engineering and Physical
      Sciences Research
      Council}: Grant number
    \href{https://gow.epsrc.ukri.org/NGBOViewGrant.aspx?GrantRef=EP/T026960/1}{EP/T026960/1, \emph{AISEC: AI Secure and
        Explainable by Construction}.}}
  \author{Robert Atkey\thanksref{msp}\thanksref{bobemail}}
  \author{Wen Kokke\thanksref{wenemail}}
  \address[msp]{%
    Mathematically Structured Programming Group\\
    Computer and Information Sciences\\
    University of Strathclyde\\
    Glasgow, Scotland, UK}
  \thanks[bobemail]%
  {Email: \href{robert.atkey@strath.ac.uk}%
    {\texttt{\normalshape robert.atkey@strath.ac.uk}}}
  \thanks[wenemail]%
  {Email: \href{me@wen.works}%
    {\texttt{\normalshape me@wen.works}}}
  \begin{abstract}
    \emph{Multiplicative-Additive System Virtual} (MAV) is a logic
    that extends Multiplicative-Additive Linear Logic with a self-dual
    non-commutative operator expressing the concept of ``before'' or
    ``sequencing''. MAV is also an extenson of the the logic Basic
    System Virtual (BV) with additives. Formulas in BV have an
    appealing reading as processes with parallel and sequential
    composition. MAV adds internal and external choice operators. BV
    and MAV are also closely related to Concurrent Kleene Algebras.
    
    Proof systems for MAV and BV are Deep Inference systems, which
    allow inference rules to be applied anywhere inside a
    structure. As with any proof system, a key question is whether
    proofs in MAV can be reduced to a normal form, removing detours
    and the introduction of structures not present in the original
    goal. In Sequent Calcluli systems, this property is referred to as
    Cut Elimination. Deep Inference systems have an analogous Cut rule
    and other rules that are not present in normalised proofs. Cut
    Elimination for Deep Inference systems has the same metatheoretic
    benefits as for Sequent Calculi systems, including consistency and
    decidability.
    
    Proofs of Cut Elimination for BV, MAV, and other Deep Inference
    systems present in the literature have relied on intrincate
    syntactic reasoning and complex termination measures.

    We present a concise semantic proof that all MAV proofs can be
    reduced to a normal form avoiding the Cut rule and other ``non
    analytic'' rules. We also develop soundness and completeness
    proofs of MAV (and BV) with respect to a class of models. We have
    mechanised all our proofs in the Agda proof assistant, which
    provides both assurance of their correctness as well as yielding
    an executable normalisation procedure. Our technique extends to
    include exponentials and the additive units.
  \end{abstract}
  \begin{keyword}
    Linear Logic, Deep Inference, Algebraic Semantics, Metatheory
  \end{keyword}
\end{frontmatter}

\section{Introduction}\label{sec:introduction}

We present an algebraic semantics and semantic proof of generalised cut-elimination for the multiplicative-additive system MAV~\cite{Horne15:mav}, which extends the basic system BV\footnote{
      BV stands for \underline{B}asic System \underline{V}irtual, owing to an early interpretation of CCS interaction as the pairwise production and annihilation of virtual particles in physics~\cite[\Fn2]{Horne15:mav}.
}~\cite{Guglielmi99:bv,Guglielmi07:sis} with the additives of multiplicative-additive linear logic~\cite[MALL]{Girard87:ll}. The proof technique also extends to the additive units and the exponentials.
Our proofs are constructive and mechanised in Agda~\cite{Agda264}.

\subsection{BV, MAV, and Deep Inference}

MAV and BV are \emph{Deep Inference} systems. Deep Inference~\cite{Guglielmi14:di} is generalisation of Gentzen's methodology for designing proof systems, which arose from Guglielmi's attempts to relate process algebra~\cite[CCS]{Milner80:CCS,Milner89:CC} to Classical Linear Logic~\cite[CLL]{Girard87:ll}.
The problem with such a relation is that, while the multiplicative connectives of Linear Logic capture parallel composition, no connective of Linear Logic captures \emph{sequential composition}.
Eventually, Guglielmi's attempts yielded BV, which extends Multiplicative Linear Logic~\cite[MLL]{Girard87:ll} with a self-dual non-commutative connective for sequential composition.
Such a connective was already present in another extension of Linear Logic, Pomset Logic~\cite{Retore97:pomset}, where it arose from the study of coherence space semantics of Linear Logic~\cite[\C8]{GirardTL89:proofs}.
Recently, Nguyễn and Stra{\ss}burger~\cite{NguyenS22:bvisnotpl,NguyenS23:complexity} showed that, while BV is similar to Pomset logic, the two are not the same, as the theorems of BV form a proper subset of the theorems of Pomset logic.
Neither BV nor Pomset logic has a sequent calculus. Tiu~\cite{Tiu06:sisii} showed that no sequent calculus can capture BV, and it is assumed that this result extends to Pomset logic.

\emph{Cut-elimination}, or \emph{admissibility of Cut}, is the fundamental property of Gentzen Sequent Calculi systems, which states that proofs using the Cut rule to introduce ``detours'' can be normalised to ones without. Crucial properties such as consistency and decidability follow from Cut-elimination. The Deep Inference analogue of Cut-elimination is the admissibility of the whole \emph{up} fragment of the calculus, which includes the Deep Inference form of Cut (which we describe in \Cref{sec:mav-syntax} below) as well as duals of most of the other rules of the calculus. Admissibility of the \emph{up} fragment has the same metatheoretic benefits for Deep Inference systems as it does for Sequent Calculus ones.

Guglielmi \cite[\S4.1]{Guglielmi14:di} proves admissibility of Cut via the Splitting Theorem, which shows that proofs of conjoined structures can be split into separate subproofs. This is proved by a detailed syntactic analysis of proofs. Horne~\cite{Horne15:mav} gave a syntactic proof of the admissibility of the up fragment for MAV, that extends Guglielmi's technique with further reasoning about the additives. The proof is quite lengthy and involves intricate syntactic reasoning and the subtle and complex termination measures.

We present an alternative proof of Cut-elimination via a semantic model. This proof avoids some of the intricacy of Horne's proof. Our proof is more robust in the presence of extensions, due to our use of standard constructions such as Day monoids, order ideals, and the Chu construction. We demonstrate this by scaling down to plain BV and up to MAV with additive units, and also including exponentials (Guglielmi and Stra{\ss}burger's System NEL \cite{Burger_2011,GuglielmiS11}).

The technique of demonstrating Cut Elimination by construction of a semantic model for MALL is due to Okada~\cite{Okada99:psc}, who shows that the Phase Space model of MALL, described by Girard~\cite[\S4.1]{Girard87:ll} and Troelstra~\cite[\C8]{Troelstra92:lll}, can be constructed from cut-free proofs.
The completeness of this model directly yields the existence of a cut-free proof for every proof constructible in the MALL sequent calculus.
The same technique was used by Abrusci~\cite{Abrusci91:psc} for non-commutative linear logic, by De, Jafarrahmani, and Saurin~\cite{De22:psc} for MALL with fixed points, and was adapted to Bunched Implications by  Frumin~\cite{Frumin22:psc}.

The Phase Semantics-based proof of cut-elimination does not easily extend to include the kind of self-dual connective present in BV and its extensions.
The Phase Space model derives duality by means of double negation with respect to the monoidal structure, which means that any connective has a derived dual. Attempts to extend the model with the non-commutative connective result in two distinct but dual non-commutative connectives, reminiscent of the non-commutative tensor and par introduced by Slavnov~\cite{Slavnov19:scmll}.

\subsection{Contribution and Content of this Paper}

To our knowledge, all prior work on the metatheory of Deep Inference systems like BV and MAV has been carried out using syntactic techniques such as rewriting with termination measures, or translations into other logics with known Cut-Elimination properties.

Our main contribution is the use of semantic techniques to derive the admissibility of identity expansion, cut, and the other \emph{up} rules of MAV. To this end, we have developed a number of results concerning the semantics of BV and MAV:
\begin{enumerate}
      \item In \Cref{sec:mav-algebras}, we propose \emph{MAV-algebras} as the algebraic counterpart of MAV. In short, an MAV-algebra is a $*$-autonomous partial order with meets, with another partially ordered monoid structure that is \emph{duoidal} with respect to the $*$-autonomous structure.
      \item Normal proofs (our name for the \emph{up} fragment) as we define them in the next section do not {\it prima facie} support all the structure of an MAV-algebra, so we define the weaker notion of \emph{MAV-frame} (\Cref{defn:mav-frame}). These play the same role as Kripke structures or frames do for modal and substructural logics. The first main technical contribution of the paper is in showing that every MAV-frame can be completed to an MAV-algebra.
      \item We then show that MAV-frames are strongly complete (in the terminology of Okada) for MAV by proving that the MAV-frame constructed from normal proofs can be used to deduce that all MAV-provable structures have normal proofs in \Cref{thm:cut-elim}. Completeness in the usual sense also follows in \Cref{thm:completeness}.
\end{enumerate}

We describe the MAV Deep Inference system in \Cref{sec:mav-syntax} and motivate the idea for readers not familiar with such systems.

As far as we are aware, the semantics of BV and its extensions have not been considered before in the literature. The crucial part of the proof is that the standard Chu construction extends to the self-dual non-commutative connective of BV and MAV (\Cref{prop:chu-self-dual} and \Cref{prop:chu-monoid-duoidal}). Our development of the semantics of MAV and BV also opens the possibility of using these logics as sound and complete systems for reasoning about the MAV-frame structures we define.

All of our proofs have been mechanised and checked in the Agda proof assistant \cite{Agda264}. We briefly discuss the mechanisation in \Cref{sec:mechanisation} and provide a hyperlinked guide to the proof relating it to our mathematical development in \Cref{sec:table-of-statments}. Aside from the benefits of checking the proof, the Agda proof is executable and yields a program for normalising proofs.

We present our proof for MAV specifically, but note that essentially the same proof applies to the subsystem BV as well. We have also extended the proof technique to include the additive units, and with exponentials (System NEL, analysed by Guglielmi and Stra{\ss}burger \cite{Burger_2011,GuglielmiS11}). We discuss further extensions in \Cref{sec:future-work}.

\section{The system MAV}\label{sec:mav-syntax}

In Deep Inference terminology, proofs operate on \emph{structures}, which simultaneously play the role of formulas and sequents in traditional Sequent Calculus systems. There are a number of different notations in the literature for the structures of BV and related systems. For familiarity's sake, we opt for a notation similar to the formulas of normal Linear Logic, albeit extended with the self-dual non-commutative connective $\vSeq$.

The structures of MAV are formed from positive and negative atoms ($\vPos\va$ and $\vNeg\va$) drawn from some set of atomic propositions, units ($\vUnit$), the non-commutative connective \emph{seq} ($\vSeq$), the multiplicative connectives \emph{tensor} and \emph{par} ($\vTens$ and $\vParr$) and additive connectives \emph{with} and \emph{plus} ($\vWith$ and $\vPlus$).
\begin{displaymath}
  \vP,\vQ,\vR,\vS
  \Coloneq \vPos\va
  \mid     \vNeg\va
  \mid     \vUnit
  \mid     \vP\vSeq \vQ
  \mid     \vP\vTens\vQ
  \mid     \vP\vParr\vQ
  \mid     \vP\vWith\vQ
  \mid     \vP\vPlus\vQ
\end{displaymath}
Duality ($\vDual\vP$) is an involutive function on structures that obeys the De Morgan laws for the multiplicative and additive connectives and preserves the self-dual connective $\vSeq$.
\begin{displaymath}
  \begin{array}{
      lcl @{\hspace{1cm}}
      lcl @{\hspace{1cm}}
      lcl @{\hspace{1cm}}
      lcl}
    \multicolumn{3}{l}{}
     & \vDual{\vUnit}
     & =
     & \vUnit
     & \vDual{\vP\vTens\vQ}
     & =
     & \vDual\vP \vParr \vDual\vQ
     & \vDual{\vP\vParr\vQ}
     & =
     & \vDual\vP \vTens \vDual\vQ
    \\
    \vDual{\vNeg\va}
     & =
     & \vPos\va
     & \vDual{\vP\vSeq \vQ}
     & =
     & \vDual\vP \vSeq  \vDual\vQ
     & \vDual{\vP\vWith\vQ}
     & =
     & \vDual\vP \vPlus \vDual\vQ
     & \vDual{\vP\vPlus\vQ}
     & =
     & \vDual\vP \vWith \vDual\vQ
  \end{array}
\end{displaymath}
Structures are considered equivalent modulo the equality $\vEquiv$, which is the smallest congruence defined by the associativity, commutativity, and identity laws that ensure that $(\vSeq,\vUnit)$ forms a monoid, and $(\vTens,\vUnit)$ and $(\vParr,\vUnit)$ form commutative monoids.
\begin{displaymath}
  \begin{array}{
      l@{\;\vEquiv\;}ll @{\hspace{1cm}}
      l@{\;\vEquiv\;}ll @{\hspace{1cm}}
      l@{\;\vEquiv\;}ll}
    \vP\vSeq\vUnit
     & \vP
     & \RuleLabel*[seq-runit]{\vSeq-Unit\textsuperscript{R}}
     &
    \vUnit\vSeq\vP
     & \vP
     & \RuleLabel*[seq-lunit]{\vSeq-Unit\textsuperscript{L}}
     &
    \vP\vSeq(\vQ\vSeq\vR)
     & (\vP\vSeq\vQ)\vSeq\vR
     & \RuleLabel*[seq-assoc]{\vSeq-Assoc}
    \\
    \vP\vTens\vUnit
     & \vP
     & \RuleLabel*[tens-unit]{\vTens-Unit}
     &
    \vP\vTens\vQ
     & \vQ\vTens\vP
     & \RuleLabel*[tens-comm]{\vTens-Comm}
     &
    \vP\vTens(\vQ\vTens\vR)
     & (\vP\vTens\vQ)\vTens\vR
     & \RuleLabel*[tens-assoc]{\vTens-Assoc}
    \\
    \vP\vParr\vUnit
     & \vP
     & \RuleLabel*[parr-unit]{\vParr-Unit}
     &
    \vP\vParr\vQ
     & \vQ\vParr\vP
     & \RuleLabel*[parr-comm]{\vParr-Comm}
     &
    \vP\vParr(\vQ\vParr\vR)
     & (\vP\vParr\vQ)\vParr\vR
     & \RuleLabel*[parr-assoc]{\vParr-Assoc}
  \end{array}
\end{displaymath}
Structure contexts are one-hole contexts over structures. Plugging ($\vC\vPlug\vP$) replaces the hole in $\vC$ with $\vP$.
\begin{displaymath}
  \vC,\vD
  \Coloneq \vEmpty
  \mid     \vC\vSeq \vQ
  \mid     \vP\vSeq \vD
  \mid     \vC\vTens\vQ
  \mid     \vP\vTens\vD
  \mid     \vC\vParr\vQ
  \mid     \vP\vParr\vD
  \mid     \vC\vWith\vQ
  \mid     \vP\vWith\vD
  \mid     \vC\vPlus\vQ
  \mid     \vP\vPlus\vD
\end{displaymath}
The inference rules of MAV are presented as a \emph{rewriting system} on structures. As this may be surprising to readers unfamiliar with deep inference, let us examine how this presentation relates to the usual presentation of linear logic.
Rule~(\ref{rule:AxiomLL}), shown below, is the axiom rule in the usual one-sided presentation of linear logic.
In the one-sided presentation, the turnstile is vestigial syntax, and can be removed.
In BV, the $\vParr$ connective plays the same role as the comma does in the antecedent of a linear logic sequent, and the $\vUnit$ plays the same role as the empty sequent, which would give us rule~(\ref{rule:AxiomBV-bad}) for BV.
However, BV's inference rules can work arbitrarily deep in the structure. (Hence, \emph{deep} inference.)
Hence, the axiom for BV is rule~(\ref{rule:AxiomBV}), where $\vC$ is a one-hole structure context.
\begin{center}
  $\inlineequation[rule:AxiomLL]{%
      \vlderivation{\vlin{}{}{\vdash\vP,\vDual\vP}{\vlhy{}}}}$
  \qquad
  $\inlineequation[rule:AxiomBV-bad]{%
      \vlderivation{\vlin{}{}{\vP\vParr\vDual\vP}{\vlhy{\vUnit}}}}$
  \qquad
  $\inlineequation[rule:AxiomBV]{%
      \vlderivation{\vlin{}{}{%
          \vC\vPlug{\vP\vParr\vDual\vP}}{%
          \vlhy{\vC\vPlug{\vUnit}}}}}$
\end{center}
Rule (\ref{rule:CutLL}) is the cut rule in the usual one-sided presentation of linear logic.
In rule (\ref{rule:CutLL}), as in any branching inference rule, the branching enforces the \emph{disjointness} of the premise derivations.
In BV, disjointness is internalised by the $\vTens$ connective.
Hence, it plays the same role as branching does in sequent calculus.
This would give us rule (\ref{rule:CutBVBad}) for BV.
However, as BV's inference rules can work arbitrarily deep in the structure, and the system contains the \cref{rule:Switch} rule, the surrounding context of $\vParr$s is unnecessary---and too restrictive. Hence, the cut for BV is rule (\ref{rule:CutBV}).
\begin{center}
  $\inlineequation[rule:CutLL]{%
      \vlderivation{\vliin{}{}{%
          \vdash\vGG,\vGG',\vGD,\vGD'}{%
          \vlhy{\vdash\vGG,\vP,\vGG'}}{%
          \vlhy{\vdash\vGD,\vDual\vP,\vGD'}}}}$
  \qquad
  $\inlineequation[rule:CutBVBad]{%
      \vlderivation{\vlin{}{}{%
          \vGG\vParr\vGG'\vParr\vGD\vParr\vGD'}{%
          \vlhy{%
            (\vGG\vParr\vP\vParr\vGG')
            \vTens
            (\vGD\vParr\vDual\vP\vParr\vGD')}}}}$
  \qquad
  $\inlineequation[rule:CutBV]{%
      \vlderivation{\vlin{}{}{%
          \vC\vPlug{\vUnit}}{%
          \vlhy{\vC\vPlug{\vP\vTens\vDual\vP}}}}}$
\end{center}
Beautifully, internalising branching makes the symmetry between the axiom and cut plain to see. In BV, to acknowledge this symmetry and the connection with Milner's CCS, the axiom and cut rules are referred to as \emph{interaction} and \emph{co-interaction}.

Proof trees are a cumbersome presentation for BV's derivations---they are convenient for branching sequent proofs, but BV derivations are sequences of structures.
\emph{Rewriting systems}, on the other hand, are a well-known and convenient notation for such derivations.
Hence, in MAV, inference rules are presented as rewrite rules.

\emph{Inference}, written $\vInferFrom$, is the smallest relation defined by the following axioms:
\newlength{\vInferFromDefnWidthOfLBL}%
\settowidth{\vInferFromDefnWidthOfLBL}{$\RuleName{AtomAxiom}$}%
\newlength{\vInferFromDefnWidthOfLHS}%
\settowidth{\vInferFromDefnWidthOfLHS}{$(\vP\vSeq\vQ)\vParr(\vR\vSeq\vS)$}%
\newlength{\vInferFromDefnWidthOfRHS}%
\settowidth{\vInferFromDefnWidthOfRHS}{$(\vP\vParr\vR)\vWith(\vQ\vParr\vR)$}%
\begin{displaymath}
  \begin{array}{lcl@{\hspace{0.5cm}}|@{\hspace{0.5cm}}l}
    \vP\vParr\vDual\vP
     & \vInferFrom
     & \vUnit
     & \makebox[\vInferFromDefnWidthOfLBL][l]{$\RuleLabel{Interact}$}
    \\
    (\vP\vTens\vQ)\vParr\vR
     & \vInferFrom
     & \vP\vTens(\vQ\vParr\vR)
     & \RuleLabel{Switch}
    \\
    \vUnit\vWith\vUnit
     & \vInferFrom
     & \vUnit
     & \RuleLabel{Tidy}
    \\
    (\vP\vSeq\vQ)\vParr(\vR\vSeq\vS)
     & \vInferFrom
     & (\vP\vParr\vR)\vSeq(\vQ\vParr\vS)
     & \RuleLabel{Sequence}
    \\
    \vP\vPlus\vQ
     & \vInferFrom
     & \vP
     & \RuleLabel{Left}
    \\
    \vP\vPlus\vQ
     & \vInferFrom
     & \vQ
     & \RuleLabel{Right}
    \\
    (\vP\vWith\vQ)\vParr\vR
     & \vInferFrom
     & (\vP\vParr\vR)\vWith(\vQ\vParr\vR)
     & \RuleLabel{External}
    \\
    (\vP\vSeq\vQ)\vWith(\vR\vSeq\vS)
     & \vInferFrom
     & (\vP\vWith\vR)\vSeq(\vQ\vWith\vS)
     & \RuleLabel{Medial}
    \\
     &
     &
     &
    \\
    \vUnit
     & \vInferFrom
     & \vP\vTens\vDual\vP
     & \RuleLabel{CoInteract}
    \\
    \vUnit
     & \vInferFrom
     & \vUnit \vPlus \vUnit
     & \RuleLabel{CoTidy}
    \\
    (\vP\vTens\vR)\vSeq(\vQ\vTens\vS)
     & \vInferFrom
     & (\vP\vSeq\vQ)\vTens(\vR\vSeq\vS)
     & \RuleLabel{CoSequence}
    \\
    \vP
     & \vInferFrom
     & \vP\vWith\vQ
     & \RuleLabel{CoLeft}
    \\
    \vQ
     & \vInferFrom
     & \vP\vWith\vQ
     & \RuleLabel{CoRight}
    \\
    (\vP \vTens \vR) \vPlus (\vQ \vTens \vR)
     & \vInferFrom
     & (\vP \vPlus \vQ) \vTens \vR
     & \RuleLabel{CoExternal}
    \\
    (\vP \vPlus \vR) \vSeq (\vQ \vPlus \vS)
     & \vInferFrom
     & (\vP \vSeq \vQ) \vPlus (\vR \vSeq \vS)
     & \RuleLabel{CoMedial}
    \\
     &
     &
     &
    \\
    \vC\vPlug\vP\vInferFrom\vC\vPlug\vQ
     & \text{if}
     & \vP\vInferFrom\vQ
     & \RuleLabel{Mono}
  \end{array}
\end{displaymath}
If $\vP\vInferFrom\vQ$, we say that $\vP$ can be inferred from $\vQ$, \ie the arrow points from conclusion to premise.

\textbf{N.B.}\ $\vP\vInferFrom\vQ$ is an inference rule, \emph{not a sequent}. In sequent calculus notation, it is $\frac{\vQ}{\vP}$, \emph{not} $\vP\vdash\vQ$.

\emph{Derivation}, written $\vInferFrom*$ is the reflexive, transitive closure of inference.
\emph{Invertible derivation}, written $\vInferFromTo*$, is the symmetric core of derivation, \ie $\vP\vInferFromTo*\vQ=\vP\vInferFrom*\vQ\cap\vQ\vInferFrom*\vP$.
\emph{Proofs} are derivations that terminate in the unit, \eg a derivation $\vP\vInferFrom*\vUnit$ is a proof of $\vP$.

The inference rules come in dual pairs. For every rule $\vP\vInferFrom\vQ$, there is a dual rule $\vDual\vQ\vInferFrom\vDual\vP$.
The exception is \cref{rule:Switch}, which is self-dual, up to commutativity.
The \cref{rule:CoInteract}, \cref{rule:CoLeft}, and \cref{rule:CoRight} rules introduce new structures going left-to-right.
Normal proofs, which we define below, avoid these \emph{synthetic} rules.

\begin{remark}
  In BV, the structural connectives are usually presented as lists, distinguished only by their brackets: $\vP\vTens\vQ$ is written as $\vls(\vP;\vQ)$; $\vP\vParr\vQ$ is written as $\vls[\vP;\vQ]$; and $\vP\vSeq\vQ$ is written as $\vls<\vP;\vQ>$.
  Inferences, derivations, and proofs are presented vertically, as (\ref{notation:vlin}), (\ref{notation:vlde}), and (\ref{notation:vlpr}), respectively.
  \begin{center}
    $\inlineequation[notation:vlin]{%
        \vlderivation{\vlin{}{}{\vP}{\vlhy{\vQ}}}}$
    \qquad
    $\inlineequation[notation:vlde]{%
        \vlderivation{\vlde{}{}{\vP}{\vlhy{\vQ}}}}$
    \qquad
    $\inlineequation[notation:vlpr]{%
        \vlderivation{\vlpr{}{}{\vP}}}$
  \end{center}

  The relation between the deductive system for BV and rewrite systems is well-known, \eg by Kahramanogullari~\cite{Kahramanogullari08:maude}, who implements proof search for several deep inference systems in Maude~\cite{ClavelDELMMQ02:maude}.
  Inferences rules are usually named with the combination of a letter and an up or down arrow, \eg \cref{rule:Interact} and \cref{rule:CoInteract} are $\iD$ and $\iU$, respectively. The exception are self-dual rules, which are named with a single letter, \eg \cref{rule:Switch} is usually named $\sw$.
\end{remark}

\begin{definition}\label[defn]{defn:normal}
  A derivation is \emph{normal} when it does not use \cref{rule:CoInteract} nor any of the other (\RuleName{Co-}) rules, and its uses of \cref{rule:Interact} are restricted to \cref{rule:AtomInteract}, as defined by the following axiom:
  \begin{displaymath}
    \begin{array}{lcl@{\hspace{0.5cm}}|@{\hspace{0.5cm}}l}
      \makebox[\vInferFromDefnWidthOfLHS][l]{$\vNeg\va\vParr\vPos\va$}
       & \vInferFrom
       & \makebox[\vInferFromDefnWidthOfRHS][l]{$\vUnit$}
       & \RuleLabel{AtomInteract}
    \end{array}
  \end{displaymath}
\end{definition}
Normal derivations avoid the use of rules that introduce new structures into proofs, and so can be termed \emph{analytic} in contrast to the need for the synthetic rules to synthesise new structures.

Our main result, \Cref{thm:cut-elim}, is that every structure that is provable in full MAV also has a normal proof. Therefore, the system with only analytic rules is complete for MAV provability. Horne \cite{Horne15:mav} proves this result via a syntactic proof involving rewriting and termination measures. In the following sections, we construct a semantic proof that normal proofs are complete for MAV. In \Cref{sec:future-work}, we report on extensions of the same technique to variants of MAV.

\section{Semantic Models for MAV}\label{sec:mav-semantics}

To prove the normalisation property for all MAV proofs, we use a
semantic technique that is akin to Okada's phase space method and to
Normalisation by Evaluation (NbE) \cite{BergerES98}. We construct a semantics of the
whole proof system from the system of normal proofs. This semantics is
constructed in such a way that after interpreting a proof, the
(existence of) a normal form can be extracted (or \emph{read back} or
\emph{reified} in NbE terminology) from the semantic proof.

To our knowledge, the semantics of MAV have not been previously
studied in the literature. The rules are an extension of MALL, so an
\emph{MAV-algebra} will partly be a $*$-autonomous partial order
(\Cref{defn:star-autonomous}) with meets and joins. The additional
structure for $\vSeq$ satisfies the conditions of a duoidal category
\cite[Definition 6.1]{Aguiar_2010} (\Cref{defn:duoidal}). We show that
MAV is sound for MAV-algebras in \Cref{thm:soundness}.

To build MAV-algebras from normal proofs, we define the weaker notion
of an \emph{MAV-frame} (\Cref{defn:mav-frame}). We show
that a combination of certain closed lower sets (\Cref{sec:closed-lower-sets}) and the Chu construction (\Cref{sec:chu}) construct an MAV-algebra from any MAV-frame. Much of
these constructions are well-known, but we have new results on lifting
the Day construction of monoids on closed lower sets and the
preservation of duoidal relationships that are required for MAV.

\subsection{Pomonoidal, $*$-autonomous, and Duoidal Structure on Partial Orders}\label{sec:mav-semantics-preliminaries}

The algebraic semantics of MAV is a collection of interacting monoids
on a partial order. We collect here the basic definitions and kinds of
interaction we will need.

\begin{definition}\label[defn]{defn:pomonoid}
  A \emph{partial order monoid (pomonoid)} $(\bullet, i)$ on a poset
  $(A, \leq)$ comprises a binary operator $\bullet : A \times A \to A$
  that is monotone in both arguments and an element $i \in A$ such
  that the usual monoid laws hold. A \emph{commutative pomonoid} is a
  pomonoid where additionally $x \bullet y = y \bullet x$.
\end{definition}

\begin{definition}\label[defn]{defn:residual}
  A commutative pomonoid $(\bullet, i)$ on a poset $(A, \leq)$ is
  \emph{residuated} if there is a function
  $\rightblackspoon : A \times A \to A$ such that $x \bullet y \leq z$
  iff $x \leq y \rightblackspoon z$.
\end{definition}

Linear logic adds a duality, or negation, to a commutative
pomonoid structure. Semantically, duality with commutativity is
captured in the definition of $*$-autonomous category, due to Barr~\cite{Barr_1979}. For our purposes, we need the partial order
analogue, also called a CL algebra by Troelstra~\cite{Troelstra92:lll}.

\begin{definition}\label[defn]{defn:star-autonomous}
  A \emph{$*$-autonomous partial order} is a structure
  $(A, \leq, \otimes, I, \lnot)$ where $(\otimes, I)$ is a pomonoid on
  $(A, \leq)$ and $\lnot : A^\op \to A$ is an anti-monotone and
  involutive operator on $A$, together satisfying
  $x \otimes y \leq \lnot z$ iff $x \leq \lnot (y \otimes z)$.  A
  $*$-autonomous partial order satisfies \emph{mix} if $\lnot I = I$.
\end{definition}

\begin{remark}
  The structure of a $*$-autonomous partial order has a number of
  immediate consequences, but we leave description of these until
  after the definition of MAV-algebra in \Cref{defn:mav-algebra}.
\end{remark}

BV and MAV extend linear logic by adding a non-commutative pomonoid
structure that interacts with the existing pomonoid via a kind of
interchange law (the \cref{rule:Sequence} rule in the proof
system). We follow Aguiar and Mahajan \cite[Definition 6.1]{Aguiar_2010} generalising
these to the case of monoids with differing units. Their terminology
is of a category having duoidal structure. We find it useful to
describe one pomonoid as being \emph{duoidal over} another to
emphasise the non-symmetric nature of the relationship, and by analogy
with one monoid distributing over another.

\begin{definition}\label[defn]{defn:duoidal}
  A pomonoid $(\bullet, i)$ is \emph{duoidal over} another pomonoid
  $(\lhd, j)$ on a partial order $(A, \leq)$ if the following
  inequalities hold:
  \begin{enumerate}
    \item $(w \lhd x) \bullet (y \lhd z) \leq (w \bullet y) \lhd (x \bullet z)$
    \item $j \bullet j \leq j$
    \item $i \leq i \lhd i$
    \item $i \leq j$
  \end{enumerate}
\end{definition}

\begin{remark}
  In the case when the two pomonoids share a common unit the last
  three conditions for a duoidal relationship are automatically
  satisfied. We can also ignore the existence of the units and just
  describe two binary operators as being duoidal. If $\bullet$ is a
  join, or $\lhd$ is a meet, then all the conditions for a duoidal
  relationship are automatically met.
\end{remark}

\subsection{MAV-algebras}
\label{sec:mav-algebras}

We define MAV-algebras as the algebraic semantics of MAV. The
definition is a direct translation of the rules of MAV into
order-theoretic language, using the definitions we have seen so far.

\begin{definition}\label[defn]{defn:mav-algebra}
  An \emph{MAV-algebra} is a structure
  $(A, \leq, \otimes, \lhd, I, \lnot)$ with the following properties:
  \begin{enumerate}
    \item $(A, \leq, \otimes, I, \lnot)$ is $*$-autonomous and satisfies \emph{mix}.
    \item $(A, \leq, \lhd, I)$ is a pomonoid.
    \item $\lhd$ is self dual: $\lnot (x \lhd y) = (\lnot x)\lhd (\lnot y)$.
    \item $(\otimes, I)$ is duoidal over $(\lhd, I)$.
    \item $(A, \leq)$ has binary meets, which we write as $x \with y$.
  \end{enumerate}
\end{definition}

\begin{proposition}\label[prop]{prop:mav-algebra-consequences}
  Let $(A, \leq, \otimes, \lhd, I, \lnot)$ be a MAV-algebra.
  \begin{enumerate}
    \item There is another commutative pomonoid structure $(\parr, I)$ on
          $(A, \leq)$, defined as
          $x \parr y = \lnot(\lnot x \otimes \lnot y)$.
    \item $(\otimes, I)$ and $(\parr, I)$ are linearly distributive
          \cite{Cockett_1999}:
          $x \otimes (y \parr z) \leq (x \otimes y) \parr z$.
    \item $(A, \leq)$ has binary joins, given by
          $x \oplus y = \lnot (\lnot x \with \lnot y)$.
    \item $\oplus$ distributes over $\otimes$:
          $x \otimes (y \oplus z) = (x \otimes y) \oplus (x \otimes z)$.
    \item $\with$ distributes over $\parr$:
          $(x \parr z) \with (y \parr z) = (x \with y) \parr z$.
    \item $\lhd$ is duoidal over $\parr$:
          $(w \parr x)\lhd (y \parr z) \leq (w\lhd y) \parr (x\lhd z)$.
    \item $\lhd$ is duoidal over $\with$:
          $(w \with x)\lhd (y \with z) \leq (w \lhd y) \with (x\lhd z)$.
    \item $\oplus$ is duoidal over $\lhd$:
          $(w \lhd x) \oplus (y \lhd z) \leq (x \oplus y) \lhd (x \oplus z)$.
  \end{enumerate}
\end{proposition}

\begin{definition}\label[defn]{defn:mav-interpretation}
  Let $\mathit{At}$ be a set of atomic propositions. Given an
  MAV-algebra $(A, \leq, \otimes, \lhd, I, \lnot)$ and valuation
  $V : \mathit{At} \to A$, define the interpretation of MAV Formulas
  as follows:
  $\sem{\vUnit} = I$,
  $\sem{\va} = V(\va)$,
  $\sem{\vDual \va} = \lnot V(\va)$,
  $\sem{\vP\vTens\vQ} = \sem{\vP} \otimes \sem{\vQ}$,
  $\sem{\vP\vParr\vQ} = \sem{\vP} \parr \sem{\vQ}$,
  $\sem{\vP\vSeq\vQ} = \sem{\vP}\lhd\sem{\vQ}$,
  $\sem{\vP\vWith\vQ} = \sem{\vP} \with \sem{\vQ}$, and
  $\sem{\vP\vPlus\vQ} = \sem{\vP} \oplus \sem{\vQ}$.
\end{definition}

\begin{lemma}\label[lem]{lem:dual-ok}
  For all $\vP$, $\sem{\vDual{\vP}} = \lnot \sem{\vP}$.
\end{lemma}

\begin{theorem}\label[thm]{thm:soundness}
  The interpretation in \Cref{defn:mav-interpretation} is sound:
  for all structures $\vP$,
  if $\vP\vInferFrom*\vUnit$,
  then $I \leq \sem{\vP}$.
\end{theorem}

\begin{proof}
  Each of the required inequalities has been established in
  \Cref{prop:mav-algebra-consequences}.
\end{proof}

\begin{remark}
  More generally, if $P \vInferFrom* Q$ in MAV, then
  $\sem{Q} \leq \sem{P}$ in an MAV-algebra. Note that the ordering is
  reversed! It will be reversed again in the definition of MAV-frame.
\end{remark}

\subsection{MAV-frames}
\label{sec:mav-frames}

To prove completeness of the normal proofs of MAV, we will construct a
particular MAV-algebra from the structures and normal proofs. Since
normal proofs do not \emph{a priori} have all the necessary structure
for an MAV-algebra, in the following sections we develop a procedure
to construct MAV-algebra from the lighter requirements of an
MAV-frame. In \Cref{sec:mav-cut-elimination} we will show that
the MAV-algebra constructed from the normal proof MAV-frame allows us
to prove that all proofs in MAV can be normalised to normal proofs.

\begin{definition}\label[defn]{defn:mav-frame}
  An \emph{MAV-frame} is a structure $(F, \leq, \parr, \lhd, i, +)$
  where $(F, \leq)$ is a partial order, $(F, \leq, \parr, i)$ is a
  commutative pomonoid, $(F, \leq, \lhd, i)$ is a pomonoid, $+$ is a
  binary monotone function on $(F, \leq)$, and these data satisfy the
  following inequalities:
  \begin{enumerate}
    \item $(w \lhd x) \parr (y \lhd z) \leq (w \parr y) \lhd (x \parr z)$
    \item $(x + y) \parr z \leq (x \parr z) + (y \parr z)$
    \item $(w \lhd x) + (y \lhd z) \leq (w + y) \lhd (x + z)$
    \item $i + i \leq i$
  \end{enumerate}
\end{definition}

\begin{remark}
  An MAV-frame is essentially two duoidal relationships and a
  distributivity law.
\end{remark}

\begin{remark}
  MAV-frames have a intuitive reading as CCS-like process algebras
  (see Milner \cite{Milner89:CC} for an introduction to CCS). If we
  assume the existence of a collection of ``action'' elements
  $a \in F$ and their duals $\overline{a} \in F$, satisfying
  $a \parr \overline{a} \leq i$, then we can read the constructs of an
  MAV-frame as parallel composition, sequential composition, and
  choice. The ordering is interpreted as a reduction relation. An
  interesting avenue for future work would be to discover to what
  extent MAV can be thought of as a logic for processes in this
  process algebra.
\end{remark}

\begin{remark}\label[remark]{remark:cka}
  MAV-frames (and MAV-algebras) are also very similar to the
  definition of a \emph{Concurrent Kleene Algebra} (CKA) due to Hoare,
  M{\"o}ller, Struth and Wehrman \cite[Definition
    4.1]{Hoare_2011}. One difference is that we do not assume that $+$
  is a join, nor do we assume the existence of infinitary
  joins. Consequently, we have no analogue of the Kleene Star. Another
  difference is that the duoidal relationship is reversed in
  MAV-frames, indicating that MAV-frames capture evolution of
  processes while MAV-algebras and CKA capture properties of
  processes.
\end{remark}

\begin{proposition}\label[prop]{prop:nmav-frame}
  The \emph{normal proof MAV-frame} \NMAV is the partial order arising
  as the quotient of the preorder formed from the structures of MAV
  and $\vP\leq\vQ$ if there is a normal derivation
  $\vP\vInferFrom*\vQ$, defined as
  $(\vSS, \vInferFrom*, \vParr, \vSeq, \vUnit, \vWith)$, where $\vSS$
  is the set of all MAV structures.  The required (in)equalities
  follow directly from the definition of $\vInferFrom*$ for normal
  proofs.
\end{proposition}

\begin{remark}
  The construction of the MAV-frame \NMAV does not use the $\vTens$
  and $\vPlus$ structure of MAV directly.  This structure is recovered
  by duality from the other connectives by the constructions in the
  rest of this section. This corresponds to the fact that the
  \textsc{Co-X} rules in MAV that we wish to show admissible are the
  ones that mention the $\vTens$ and $\vPlus$ connectives, with the
  exception of \cref{rule:Switch}, which has a special role to play in
  \Cref{prop:embedding-sem} in mediating interaction.
\end{remark}

\subsection{Constructing MAV-algebras from MAV-frames}

We construct MAV-algebras from MAV-frames in a three step process. In
\Cref{sec:lower-sets}, we use lower sets and the Day construction to
add meets, joins and residuals for pomonoids to a partial order. This
construction creates joins freely, so we restrict to $+$-closed lower
sets (\ie, order ideals, but not necessarily over a
$\lor$-semilattice) in \Cref{sec:closed-lower-sets} to turn the $+$
operation in MAV-frames into joins. Restricting to $+$-closed lower
sets separates the Day construction of pomonoids into two separate
cases, depending on how the original pomonoid interacts with
$+$. Finally, we create the $*$-autonomous structure using the Chu
construction in \Cref{sec:chu}. The necessary duoidal structure is
maintained through each construction.

\subsubsection{Lower Sets and Day pomonoids}
\label{sec:lower-sets}

\begin{definition}\label[defn]{defn:lower-set}
  Given a partial order $(A, \leq)$, the set of lower sets
  $\LowerSet{A}$ consists of subsets $F \subseteq A$ that are
  down-closed: $x \in F$ and $y \leq x$ implies $y \in F$. Lower sets
  are ordered by inclusion. Define the embedding
  $\eta : A \to \LowerSet{A}$ as $\eta(x) = \{ y \mid y \leq x \}$.
\end{definition}

\begin{proposition}\label[defn]{defn:lower-set-embed-monotone}
  For any $(A, \leq)$, the function $\eta$ is monotone, and
  $(\LowerSet{A}, \subseteq)$ has meets and joins given by
  intersection and union respectively.
\end{proposition}

\begin{proposition}\label[prop]{prop:day-construction}
  If $(\bullet, i)$ is a pomonoid on $(A, \leq)$, then there is a
  corresponding \emph{Day pomonoid} $(\Day{\bullet}, \Day{i})$ on
  $\LowerSet{A}$ defined as
  $F \Day\bullet G = \{ z \mid z \leq x \bullet y, x \in F, y \in G
    \}$ and $\Day{i} = \eta(i)$. Moreover:
  \begin{enumerate}
    \item If $(\bullet, i)$ is a commutative pomonoid, then so is
          $(\Day{\bullet}, \Day{i})$.
    \item $(\Day{\bullet}, \Day{i})$ has left and right residuals, which
          coincide when it is commutative. We will only be interested in
          residuals for commutative pomonoids, which we write as
          $F \rightblackspoon G$.
    \item The embedding preserves the monoid:
          $\eta(x \bullet y) = \eta(x) \Day\bullet \eta(y)$.
  \end{enumerate}
\end{proposition}

\begin{remark}
  \Cref{prop:day-construction} is the Day monoidal product on functor
  categories \cite{Day_1970} restricted to the case of partial orders
  and lower sets.
\end{remark}

\begin{remark}
  When $(A, \leq)$ is an MAV-frame, \Cref{prop:day-construction} gives us two pomonoids
  $(\Day{\parr}, \Day{I})$ and $(\Day{\lhd}, \Day{I})$ on
  $\LowerSet{A}$. Moreover, the next proposition states that the
  duoidal relationship between these monoids is preserved by the Day
  construction:
\end{remark}

\begin{proposition}\label[prop]{prop:lower-set-duoidal}
  If $(\bullet, i)$ is duoidal over $(\lhd, j)$ then
  $(\Day{\bullet}, \Day{i})$ is duoidal over $(\Day{\lhd}, \Day{j})$.
\end{proposition}

\subsubsection{$+$-closed Lower Sets}
\label{sec:closed-lower-sets}

In the Phase Semantics, structures are interpreted as elements that are
fixed points of a closure operator defined by a double negation with
respect to the monoid on the original frame. This closure operator
generates a partial order of ``facts'' whose meets and joins exactly
correspond to the syntactic ones when the original monoid is derived
from the proofs. As we mentioned in the introduction, the presence of
a self-dual operator in BV means that we cannot use double negation
closure, and we have to proceed more deliberately to preserve
join-like structure in an MAV-frame when building MAV-algebras. We do
this by defining $+$-closed lower sets as those that are closed under
finite $+$-combinations of their elements. This leads to a closure
operator on lower sets that allows us to immediately deduce that
$+$-closed lower sets form a lattice. We also preserve the
Day-pomonoids from lower sets, but in two different ways, depending on
how the original pomonoid interacts with $+$. In \Cref{prop:closed-monoid-distrib} we handle pomonoids that distribute
over $+$, and in \Cref{prop:closed-monoid-duoidal} we
handle pomonoids that are duoidal under $+$. We need these
constructions to lift the $\parr$ and $\lhd$ pomonoids from
MAV-frames to $+$-closed lower sets. Finally in this section, we show
that duoidal structure on lower sets from
\Cref{prop:lower-set-duoidal} is preserved in $+$-closed lower sets.

For this section, we assume that $(A, \leq)$ is a partial order with a
monotone binary operation $+ : A \times A \to A$ (we do not assume
that $+$ is a join or even a pomonoid.)

\begin{definition}\label[defn]{defn:closed-lower-set}
  A lower set $F \in \LowerSet{A}$ is \emph{$+$-closed} if $x \in F$
  and $y \in F$ imply $x + y \in F$. $+$-closed lower sets are ordered
  by subset inclusion and form a partial order
  $(\ClosedLowerSet{A}, \subseteq)$.
\end{definition}

\begin{proposition}
  Let $U : \ClosedLowerSet{A} \to \LowerSet{A}$ be the ``forgetful''
  function that forgets the $+$-closed property. There is a monotone
  function $\alpha : \LowerSet{A} \to \ClosedLowerSet{A}$ such that
  for all $F \in \ClosedLowerSet{A}$, $\alpha(U F) = F$ and for all
  $F \in \LowerSet{A}$, $F \subseteq U (\alpha F)$.
\end{proposition}

\begin{proof}
  To define $\alpha$, we close lower sets under all
  $+$-combinations. To this end, for $F \in \LowerSet{A}$, define
  $\mathrm{ctxt}(F)$, the set of all $+$-combinations of $F$
  inductively built from constructors
  $\mathsf{leaf} : F \to \mathrm{ctxt}(F)$ and
  $\mathsf{node} : \mathrm{ctxt}(F) \times \mathrm{ctxt}(F) \to
    \mathrm{ctxt}(F)$. We define the \emph{sum} of a context as
  $\mathit{sum}(\mathsf{leaf}~x) = x$ and
  $\mathit{sum}(\mathsf{node}(c,d)) = \mathit{sum}(c) +
    \mathit{sum}(d)$. Now define:
  $\alpha(F) = \{ x \mid c \in \mathrm{ctxt}(F), x \leq
    \mathit{sum}(c) \}$. This is $+$-closed, by taking the
  $\mathsf{node}$ combination of contexts. $\alpha \circ U$ is
  idempotent because $\alpha$ does not introduce any elements to lower
  sets that are already closed. For arbitrary lower sets $F$,
  $F \subseteq U(\alpha F)$ by the $\mathsf{leaf}$ constructor.
\end{proof}

\begin{definition}
  Define the embedding $\eta^+ : A \to \ClosedLowerSet{A}$ as
  $\eta^+(x) = \alpha(\eta(x))$.
\end{definition}

\begin{remark}
  By this proposition, $U \circ \alpha$ is a closure operator on
  $\LowerSet{A}$ \cite{Davey_2002}, and the closed elements are
  those of $\ClosedLowerSet{A}$. The next proposition is standard for
  showing that meets and joins exist on the closed elements for some
  closure operator.
\end{remark}

\begin{proposition}
  $(\ClosedLowerSet{A}, \subseteq)$ has all meets and joins. In the
  binary case, meets are defined by intersection and joins are defined
  by $F \lor G = \alpha (U F \cup U G)$.
\end{proposition}

\begin{proposition}\label[prop]{prop:closed-eta-preserve-joins}
  $\eta^+(x + y) \subseteq \eta^+(x) \lor \eta^+(y)$.
\end{proposition}

\begin{remark}
  \Cref{prop:closed-eta-preserve-joins} is the reason for
  requiring $+$-closure. This property will allow us to prove the
  crucial embedding property for all structures in \Cref{sec:mav-cut-elimination}.
\end{remark}

\begin{proposition}\label[prop]{prop:closed-monoid-distrib}
  For a commutative pomonoid $(\bullet, i)$ on $(A, \leq)$ that
  distributes over $+$
  (i.e., $(x + y) \bullet z \leq (x \bullet z) + (y \bullet z)$ holds), we have that
  $F \ClosedDay{\bullet} G = \alpha(U F \Day{\bullet} U G)$ and
  $\ClosedDay{j} = \alpha(\Day{j})$ define a residuated commutative
  pomonoid on $\ClosedLowerSet{A}$. Moreover,
  $\eta^+(x \bullet y) = \eta^+(x) \ClosedDay{\bullet} \eta^+(y)$.
\end{proposition}

\begin{proof}
  Define an operation
  $\bullet^c : \mathrm{ctxt}(F) \times \mathrm{ctxt}(G) \to
    \mathrm{ctxt}(F \Day{\bullet} G)$ that ``multiplies'' two trees,
  such that
  $\mathit{sum}(c) \bullet \mathit{sum}(d) \leq \mathit{sum}(c
    \bullet^c d)$, using the distributivity. This allows us to show that
  $\alpha$ preserves the monoid operation:
  $\alpha F \ClosedDay{\bullet} \alpha G = \alpha (F \Day{\bullet}
    G)$. With this, we can show that the monotonicity, associativity,
  unit, and commutativity properties of $\Day{\bullet}$ transfer over
  to $\ClosedDay{\bullet}$. The definition of the residual from lower
  sets is already $+$-closed, by distributivity.
\end{proof}

\begin{proposition}\label[prop]{prop:closed-monoid-duoidal}
  For a pomonoid $(\lhd, j)$ on $(A, \leq)$, if this satisfies
  $(w \lhd x) + (y \lhd z) \leq (w + y) \lhd (x + z)$ then the Day
  construction
  $F \Day{\lhd} G = \{ z \mid z \leq x \lhd y, x \in F, y \in G \}$ on
  lower sets is $+$-closed when $F$ and $G$ are. We write
  $F \ClosedDay{\lhd} G$ to indicate when we mean this construction as
  an operation on $+$-closed lower sets. If $j + j \leq j$, then the
  Day unit $\Day{j} = \eta(j)$ is also closed and we write it as
  $\ClosedDay{j} \in \ClosedLowerSet{A}$. Together
  $(\ClosedDay{\lhd}, \ClosedDay{j})$ form a pomonoid on
  $(\ClosedLowerSet{A}, \subseteq)$. Moreover,
  $\eta^+(x \lhd y) \leq \eta^+(x) \ClosedDay{\lhd} \eta^+(y)$.
\end{proposition}

\begin{proof}
  Since $+$ is duoidal over $(\lhd, j)$, the Day monoid $\Day{\lhd}$
  is automatically $+$-closed by calculation. The monoid structure
  directly transfers. Similarly, $\eta(j)$ is automatically $+$-closed
  since $j + j \leq j$.
\end{proof}

\begin{remark}
  Generalising the situation for the unit $j$ in \Cref{prop:closed-monoid-duoidal}, $\eta(x)$ is closed for any $x$
  such that $x + x \leq x$. Note that if $+$ were a join on
  $(A, \leq)$, then this would automatically be satisfied.
\end{remark}

\begin{remark}
  We have used the same decoration $\ClosedDay{\bullet}$ and
  $\ClosedDay{\lhd}$ for two separate constructions of pomonoids on
  $+$-closed lower sets. We will be careful to distinguish which we
  mean: in our present application, a symmetric operator like
  $\bullet$ will distribute over $+$ and so $\ClosedDay{\bullet}$ will
  be constructed by \Cref{prop:closed-monoid-distrib}; and
  a non-symmetric operator like $\lhd$ will be duoidal under $+$ and
  so $\ClosedDay{\lhd}$ will be constructed by \Cref{prop:closed-monoid-duoidal}.
\end{remark}

\begin{remark}
  If we have two pomonoids on $(A, \leq)$ that share a unit, then the
  two constructions of units in \Cref{prop:closed-monoid-distrib,prop:closed-monoid-duoidal} will yield the same element of
  $\ClosedLowerSet{A}$.
\end{remark}

\begin{proposition}
  If $(\bullet, i)$ is duoidal over $(\lhd, j)$ on $(A, \leq)$, and
  $(\bullet, i)$ distributes over $+$ (as in \Cref{prop:closed-monoid-distrib}) and $+$ is duoidal over
  $(\lhd, j)$ (as in \Cref{prop:closed-monoid-duoidal}),
  then $(\ClosedDay{\bullet}, \ClosedDay{i})$ is duoidal over
  $(\ClosedDay{\lhd}, \ClosedDay{j})$ on
  $(\ClosedLowerSet{A}, \subseteq)$.
\end{proposition}

\begin{proof}
  The duoidal relationship established in \Cref{prop:lower-set-duoidal} carries over thanks to the properties
  of $\alpha$ and $U$. The fact that
  $\ClosedDay{i} \subseteq \ClosedDay{j}$ relies on the condition
  $j + j \leq j$ to collapse $+$-contexts of $j$s.
\end{proof}

\subsubsection{Chu Construction}
\label{sec:chu}

To construct suitable MAV-algebras, we use the partial order version
of the Chu construction \cite[Appendix by Po-Hsiang
  Chu]{Barr_1979}. The Chu construction builds $*$-autonomous categories
from symmetric monoidal closed categories with pullbacks. In the
partial order case, the requirement for pullbacks simplifies to binary
meets. For this section, we let
$(A, \leq, \bullet, i, \rightblackspoon)$ be a partial order with a
residuated pomonoid structure and all binary meets.

\begin{definition}\label[defn]{defn:chu}
  Let $k$ be an element of $A$. $\Chu(A, k)$ is the partial order with
  elements pairs $(a^+, a^-)$ such that $a^+ \bullet a^- \leq k$, with
  ordering $(a^+,a^-) \sqsubseteq (b^+, b^-)$ when $a^+ \leq b^+$ and
  $b^- \leq a^-$. There is a monotone embedding function
  $\ChuEmbed : A \to \Chu(A, k)$ defined as
  $\ChuEmbed(x) = (x, x \rightblackspoon k)$.
\end{definition}

\begin{proposition}
  $(\Chu(A, k), \sqsubseteq)$ has a $*$-autonomous structure defined
  as:
  \begin{displaymath}
    (a^+, a^-) \otimes (b^+, b^-) = (a^+ \bullet b^+, (b^+ \rightblackspoon a^-) \land (a^+ \rightblackspoon b^-))
    \qquad
    I = (i, k)
    \qquad
    \lnot(a^+,a^-) = (a^-, a^+)
  \end{displaymath}
  Moreover,
  $\ChuEmbed(x \bullet y) = \ChuEmbed(x) \otimes \ChuEmbed(y)$ and
  $\ChuEmbed(i) = I$.
\end{proposition}

\begin{remark}
  If we choose $k = i$, then $(\Chu(A, i), \sqsubseteq)$ has
  $*$-autonomous structure that satisfies \emph{mix}.
\end{remark}

\begin{proposition}\label[prop]{prop:chu-meets}
  If $A$ has binary joins, then $(\Chu(A, k), \sqsubseteq)$ has
  binary meets, given by
  $(a^+,a^-) \with (b^+,b^-) = (a^+ \land b^+, a^- \lor b^-)$.
\end{proposition}

\begin{remark}
  Since $(\Chu(A, k), \sqsubseteq, \otimes, I, \lnot)$ is a
  $*$-autonomous partial order, then \Cref{prop:chu-meets}
  also means that $\Chu(A, k)$ has all binary joins, with
  $(a^+, a^-) \oplus (b^+, b^-) = (a^+ \lor b^+, a^- \land b^-)$.
\end{remark}

We now turn to the self-dual duoidal structure required to interpret
the $\vSeq$ connective. First we transfer pomonoids from $(A, \leq)$
to self-dual pomonoids on $(\Chu(A, k), \sqsubseteq)$ provided they
interact well with $k$:
\begin{proposition}\label[prop]{prop:chu-self-dual}
  Let $(\lhd, j)$ be a pomonoid on $(A, \leq)$ such that
  $(\bullet, i)$ is duoidal over $k \lhd k \leq k$ and $j \leq k$,
  then $x \lhd y = (x^+ \lhd y^+, x^- \lhd y^-)$ and $J = (j, j)$ form
  a self-dual pomonoid on $\Chu(A, k)$.
\end{proposition}

\begin{proof}
  $x\lhd y$ is well-defined because
  $(x^+ \lhd y^+) \bullet (x^- \lhd y^-) \leq (x^+ \bullet x^-) \lhd
    (y^+ \bullet y^-) \leq k \lhd k \leq k$. $J$ is well defined because
  $j \bullet j \leq j \leq k$. The pomonoid laws all transfer directly.
\end{proof}

\begin{remark}
  When $k = j$, the two conditions in the proposition are
  automatically satisfied. Moreover, if $k = i = j$, then not only
  does the $*$-autonomous structure satisfy \emph{mix}, but we also
  have $I =J$.
\end{remark}

Finally, we need to show that if $(\bullet, j)$ is duoidal over
$(\lhd, j)$, then their Chu counterparts are in the same
relationship. Due to the use of residuals in the definition of
$\otimes$, we need the following fact about duoidal residuated
pomonoids:

\begin{lemma}\label[lem]{lem:duoidal-residual}
  If $(\bullet, j)$ is duoidal over $(\lhd, j)$ in a partial order
  $(A, \leq)$ and $(\bullet, j)$ has a residual $\rightblackspoon$,
  then
  $(w \rightblackspoon x) \lhd (y \rightblackspoon z) \leq (w \lhd y)
    \rightblackspoon (x \lhd z)$.
\end{lemma}

\begin{remark}
  \Cref{lem:duoidal-residual} is in some sense the
  ``intuitionistic'' version of the duoidal relationship for $\parr$
  arising as the dual of that for $\otimes$ in a $*$-autonomous
  partial order, as we saw in \Cref{prop:mav-algebra-consequences}.
\end{remark}

\begin{proposition}\label[prop]{prop:chu-monoid-duoidal}
  If $(\bullet, i)$ is duoidal over $(\lhd, j)$ on $(A, \leq)$, and
  $(\lhd, j)$ satisfies the conditions of \Cref{prop:chu-self-dual}, then $(\otimes, I)$ and $(\lhd, J)$ are in a
  duoidal relationship on $\Chu(A, k)$.
\end{proposition}

\begin{proof}
  For the positive half of the Chu construction, this is a direct
  consequence of the duoidal relationship. For the negative half, we
  use \Cref{lem:duoidal-residual} and the fact that meets are
  always duoidal.
\end{proof}

\subsubsection{Construction of MAV-algebras from MAV-frames}
\label{sec:algebra-from-frame}

The propositions in the preceeding three sections together prove that
every MAV-frame yields an MAV-algebra:

\begin{theorem}\label[thm]{thm:algebra-from-frame}
  If $(F, \leq, \parr, \lhd, i, +)$ is an MAV-frame, then
  $(\Chu(\ClosedLowerSet{F}, \ClosedDay{i}), \sqsubseteq)$ has the
  structure of an MAV-algebra.
\end{theorem}

With this theorem we can define a notion of validity in MAV in terms
of truth in all MAV-frame generated algebras. By \Cref{thm:soundness},
MAV is sound for this notion of validity:

\begin{theorem}\label[thm]{thm:mavframe-soundness}
  MAV is sound for the MAV-frame semantics: if $P \vInferFrom* \vUnit$
  then for all MAV-frames $F$, $I \sqsubseteq \sem{P}$ in
  $(\Chu(\ClosedLowerSet{F}, \ClosedDay{i}), \sqsubseteq)$.
\end{theorem}

\section{Semantic Cut-Elimination and Proof Normalisation}
\label{sec:mav-cut-elimination}

Let $\Chu(\ClosedLowerSet{\NMAV}, \ClosedDay{\vUnit})$ be the MAV-algebra
constructed (\Cref{thm:algebra-from-frame}) from the normal
proof MAV-frame (\Cref{prop:nmav-frame}), where elements
are positive/negative pairs of $+$-closed lower sets of structures. We
define the valuation of atoms as
$V(\va) = \ChuEmbed(\ClosedLowerEmbed(\vDual\va))$. By
\Cref{thm:soundness}, we have an interpretation of MAV structures
$\sem{\vP}$ such that if $\vP\vInferFrom*\vUnit$, then
$I \sqsubseteq \sem{\vP}$. We now prove our main proposition about this
interpretation in the MAV-algebra derived from the MAV-frame of normal
proofs \NMAV{} that will allow us to derive the admissibility of all
the non-normal proof rules of MAV.

\begin{proposition}\label[prop]{prop:embedding-sem}
  For all structures $\vP$, $\sem{\vP} \sqsubseteq \lnot (\ChuEmbed(\ClosedLowerEmbed(P)))$.
\end{proposition}

\begin{proof}
  By \Cref{defn:chu}, this statement comprises two inclusions
  between pairs of $+$-closed lower sets:
  \begin{enumerate}
    \item $\ClosedLowerEmbed(P) \subseteq \sem{\vP}^-$
    \item $\sem{\vP}^+ \subseteq \ClosedLowerEmbed(P) \rightblackspoon^+ \ClosedDay{I}$
  \end{enumerate}
  We prove the second assuming the first. It suffices to prove that
  $\sem{\vP}^+ \ClosedDay{\bullet} \ClosedLowerEmbed(P) \subseteq
    \ClosedDay{I}$, which follows from the first part and the property
  of all Chu-elements that
  $\sem{\vP}^+ \ClosedDay{\bullet} \sem{\vP}^- \subseteq \ClosedDay{I}$.

  We prove the first part by induction on $\vP$.
  In the cases when $\vP = \vUnit$ or $\vP = \vDual\va$, we already have
  $\sem{\vP}^- = \ClosedLowerEmbed(P)$. When $\vP = \va$, we have
  $\sem{\va}^- = \ClosedLowerEmbed(\vDual\va)
    \rightblackspoon^+ \ClosedDay{I}$. It suffices to prove that
  $\ClosedLowerEmbed(\va) \ClosedDay{\bullet}
    \ClosedLowerEmbed(\vDual\va) \subseteq \ClosedDay{I}$, which
  follows from the preservation of monoid operations by
  $\ClosedLowerEmbed$ and the \cref{rule:AtomInteract} rule.

  When $\vP = \vQ\vParr\vR$, $\vQ\vWith\vR$, or $\vQ\vSeq\vR$, the result
  follows from preservation of the corresponding monoid structure by
  $\ClosedLowerEmbed$. For example,
  $\ClosedLowerEmbed(\vQ\vParr\vR) \subseteq \ClosedLowerEmbed(\vQ)
    \ClosedDay{\bullet} \ClosedLowerEmbed(\vR) \subseteq \sem{\vQ}^-
    \ClosedDay{\bullet} \sem{\vR}^- = \sem{\vQ\vParr\vR}^-$.

  When $\vP = \vQ\vPlus\vR$, we have
  $\ClosedLowerEmbed(\vQ\vPlus\vR) \subseteq \ClosedLowerEmbed(\vQ)$ and
  $\ClosedLowerEmbed(\vQ\vPlus\vR) \subseteq \ClosedLowerEmbed(\vR)$, by the
  \cref{rule:Left} and \cref{rule:Right} rules. Therefore,
  $\ClosedLowerEmbed(\vQ\vPlus\vR) \subseteq \ClosedLowerEmbed(\vQ) \lor
    \ClosedLowerEmbed(\vR) \subseteq \sem{\vQ}^- \lor \sem{\vR}^- = \sem{\vQ\vPlus\vR}^-$.

  When $\vP = \vQ\vTens\vR$, we have
  $\sem{\vQ\vTens\vR}^- = (\sem{\vQ}^+ \rightblackspoon \sem{\vR}^-) \land
    (\sem{\vR}^+ \rightblackspoon \sem{\vQ}^-)$. We prove inclusion in the
  left-hand side, the right-hand side is similar. The key property we
  need to prove is:
  \begin{equation}\label{eq:interaction-leq}
    \ClosedLowerEmbed(\vQ\vTens\vR) \ClosedDay{\bullet} (\ClosedLowerEmbed(Q) \rightblackspoon^+ \ClosedDay{I})
    \subseteq
    \ClosedLowerEmbed(R)
  \end{equation}
  Using the monoidality and monotonicity of $\va$, this inclusion
  is implied by the following inclusion in $\LowerSet{\NMAV}$:
  \begin{displaymath}
    \LowerEmbed(\vQ\vTens\vR) \Day{\bullet} (U(\ClosedLowerEmbed(Q)) \rightblackspoon \Day{I})
    \subseteq
    \LowerEmbed(R)
  \end{displaymath}
  which follows from the \cref{rule:Switch} rule of MAV and
  calculation. Using \ref{eq:interaction-leq}, and inclusion (ii)
  above, we can prove the inequality we need:
  \begin{displaymath}
    \ClosedLowerEmbed(\vQ\vTens\vR) \ClosedDay{\bullet} \sem{\vQ}^+
    \subseteq
    \ClosedLowerEmbed(\vQ\vTens\vR) \ClosedDay{\bullet} (\ClosedLowerEmbed(Q) \rightblackspoon^+ \ClosedDay{I})
    \subseteq
    \ClosedLowerEmbed(R)
    \subseteq
    \sem{\vR}^-
  \end{displaymath}
  Using the residuation property of $\ClosedDay{\bullet}$ we can conclude.
\end{proof}

\begin{theorem}\label[thm]{thm:cut-elim}
  If $\vP\vInferFrom*\vUnit$ in MAV, then there is a normal proof $\vP\vInferFrom*\vUnit$.
\end{theorem}

\begin{proof}
  By \Cref{thm:soundness}, $\vP\vInferFrom*\vUnit$ in MAV
  implies $I \sqsubseteq \sem{\vP}$. Combined with \Cref{prop:embedding-sem}, we have
  $I \sqsubseteq \lnot \ChuEmbed(\LowerEmbed(\vP)$. By
  \Cref{defn:chu} of the ordering of Chu elements, we have
  $\ClosedLowerEmbed(\vP) \subseteq \ClosedDay{I}$. Since
  $P \in \ClosedLowerEmbed(\vP)$, we have $\vP \in \ClosedDay{I}$, which
  by definition means that there is a normal proof
  $\vP\vInferFrom*\vUnit$.
\end{proof}

Another consequence of \Cref{prop:embedding-sem} is that
MAV is complete for the MAV-frame semantics:

\begin{theorem}\label[thm]{thm:completeness}
  MAV is complete for the MAV-frame semantics: if, for all MAV-frames
  $F$, $\sem{\vUnit} \sqsubseteq \sem{\vP}$ in
  $(\Chu(\ClosedLowerSet{F}, \ClosedDay{I}), \sqsubseteq)$, then
  $\vP\vInferFrom*\vUnit$.
\end{theorem}

\section{Mechanisation in Agda}
\label{sec:mechanisation}
We formalised the proofs in the paper in Agda~\cite{Agda264}.
The source code is available at the following URL:
\begin{center}
  \url{https://github.com/bobatkey/semantic-cut-elimination}
\end{center}
Furthermore, a hyperlinked HTML rendition of the source can be browsed at the following URL:
\begin{center}
  \url{\AgdaBaseUrl}
\end{center}
In \Cref{sec:table-of-statments}, we provide a guide to the mechanisation relating the Definitions, Propositions, and Theorems in the previous sections to the Agda definitions in the mechanisation.
The formalisation uses setoids to represent sets, and reuses definitions from the Agda Standard Library~\cite{AgdaStdlib20} where appropriate.

We did not attempt to formalise Horne's syntactic proof of generalised cut-elimination directly. We suspect that this would likely be quite involved, due to the widespread and implicit use of syntactic equalities when manipulating structures, as well as the construction of the relevant termination measures. We found that the semantic constructions were relatively straightforward to formalise in Agda.

In addition to increasing the confidence in our results, a key benefit of the formalised proof in a proof assistant for constructive proof such as Agda is that the proof normalisation procedure defined by \Cref{thm:cut-elim} is executable.
As an example, we have normalised the one-step proof below:
\begin{displaymath}
  ((\vUnit \vPlus \vUnit) \vSeq (\vUnit \vWith \vUnit))
  \vParr
  ((\vUnit \vWith \vUnit) \vSeq (\vUnit \vPlus \vUnit))
  \xlongrightarrow{\cref{rule:Interact}}
  \vUnit
\end{displaymath}
The proof normalises to a 38-step normal proof, of which 9 are inference steps, and the remainder are (sometimes spurious) equalities.
The example can be found at the following URL:
\begin{center}
  \url{\AgdaBaseUrl/MAV.Example.html}
\end{center}

\section{Extensions and Future Work}\label{sec:future-work}
We have presented a semantic proof of generalised Cut elimination for the Multiplicative-Additive System Virtual (MAV), which reduces Horne's approximately 41 page proof to a 7 page proof.

Our proof technique is modular, and can be adapted with relative ease to a variety of related systems.
To evidence this claim, we have adapted our Agda formalisation to prove generalised cut elimination to the following systems:
\begin{description}
  \item[BV]
        The basic system does not have the additives (\ie $\vWith$ and $\vPlus$).
        The proof is a straightforward restriction of our proof for MAV,
        but only relies on lower sets, rather than $+$-closed lower sets.
        The source code of the proof is available at the following URL:
        \begin{center}
          \url{\AgdaBaseUrl/BV.CutElim.html}
        \end{center}
  \item[MAUV]
        The multiplicative-additive-unital system adds the additive units (\ie $\top$ and $\mathbf{0}$, using Girard's notation).
        The proof is a straightforward extension of our proof for MAV, and requires the use of lower sets which are $0$-closed as well as $+$-closed.
        The source code of the proof is available at the following URL:
        \begin{center}
          \url{\AgdaBaseUrl/MAUV.CutElim.html}
        \end{center}
  \item[NEL]
        Non-commutative exponential logic~\cite{GuglielmiS11} extends BV with the exponentials (\ie $!$ and $?$, using Girard's notation).
        The proof is a straightforward extension of our proof for MAV, and requires that we:
        \begin{enumerate*}
          \item add \emph{near-exponentials}, which do not satisfy monoidality, to BV-frames;
          \item construct exponentials on lower sets, which adds monoidality;
          \item construct exponentials on Chu spaces, which adds duality; and
          \item extend the main proposition (\Cref{prop:embedding-sem}) to account for the new connectives, which requires lemmas similar to those for $\vTens$ and $\vParr$.
        \end{enumerate*}
        The source code of the proof is available at the following URL:
        \begin{center}
          \url{\AgdaBaseUrl/NEL.CutElim.html}
        \end{center}
  \item[MAUVE]
    The combination of all the previous systems is MAUVE: \textbf{M}ultiplicative-\textbf{A}dditive-\textbf{U}nital System \textbf{V}irtual with \textbf{E}xponentials. The construction of exponentials on lower sets extends to $0$- and $+$-closed ideals. The source code of the proof is available at the following URL:
    \begin{center}
      \url{\AgdaBaseUrl/MAUVE.CutElim.html}
    \end{center}
\end{description}
The work opens up several paths for future work.
The theory developed here for lifting Day pomonoids to $+$-closed lower sets enables alternative Cut-elimination proofs for other substructural logics, such as MALL and Bunched Implications. (Okada's technique has already been adapted to Bunched Implications by Frumin~\cite{Frumin22:psc}.)
We find the technique of using $+$-closed lower sets more revealing in how the join structure is preserved than the use of impredicative closure operators in Frumin's proof or double negation closure in Okada's.

We plan to investigate extensions of MAV with a Kleene Star operator, which can be seen as the exponential for the $\vSeq$ connective. Adding a Kleene Star would tighten the connection with Concurrent Kleene Algebras we highlighted in~\Cref{remark:cka}. It would interesting to see to what extent MAV can be seen as a logic for processes represented as elements of MAV-frames.
More generally, we plan to investigate fixpoint operators following Baelde~\cite{Baelde12} and De, Jafarrahmani and Saurin~\cite{De22:psc}.
The latters' use of Okada's technique is not compatible with Agda's type theory, as it relies on impredicativity to construct fixpoints with the double negation closure. We believe that our more direct predicative technique will be able to use Agda's inductive and coinductive types.

Lastly, we plan to extend our semantics of BV and MAV to a category-theoretic semantics that considers equalities between proofs as well as provability.
Such a semantics ought to be useful for treating MAV as a session-typed language, as considered by Ciobanu and Horne~\cite{Ciobanu_2016}.
The necessary analogue of MAV-algebras has already been investigated by Blute, Panangaden, and Slavnov~\cite{Blute_2010} as BV-categories, which are Aguiar and Mahajan's 2-monoidal (or duoidal) categories~\cite{Aguiar_2010} extended with duality.
The key task will be to categorify the constructions in this paper to show how the categorical analogue of MAV-frames induces MAV-categories.

\begin{acknowledgement}
  We would like to thank Ross Horne for helpful comments and pointers
  to related work. This work was funded by the
  \href{https://www.ukri.org/about-us/epsrc/}{UKRI Engineering and
    Physical Sciences Research Council (EPSRC)}, grant number
  \href{https://gow.epsrc.ukri.org/NGBOViewGrant.aspx?GrantRef=EP/T026960/1}{EP/T026960/1}
  ``\emph{AISEC: AI Secure and Explainable by Construction}''.
\end{acknowledgement}

\bibliographystyle{entics}
\bibliography{paper}

\appendix

\section{Statements and corresponding Agda definitions}
\label{sec:table-of-statments}

The table below lists the main definitions, propositions, and theorems in the paper together with the qualified name of the corresponding Agda definition.
The names of the Agda definitions are hyperlinked to the HTML rendition of the source code, which is available at the following URL:
\begin{center}
  \url{\AgdaBaseUrl}
\end{center}
If the version of the paper you are reading does not support hyperlinks, you can reconstruct the URL to any definition by taking the URL above, adding a \texttt{/}, the module path followed by \texttt{.html}, and the definition name prefixed with a \texttt{\#}.
For instance, the URL to the definition of \texttt{\`{}¬\_} is:
\begin{center}
  \texttt{\AgdaBaseUrl/MAV.Structure.html\#\`{}¬\_}
\end{center}
The table has three columns.
The first gives an intuitive name to the definition, proposition, or theorem statement.
The second links to its statement in the paper, if applicable.
The third gives the name of the corresponding Agda definition, and links to its statement online.
\vspace*{2ex}
\begin{longtable}[c]{lll}
  \caption{Statements and corresponding Agda definitions.}
  \label[tbl]{tbl:index}
  \\
  \longtableheader{\Cref{sec:mav-syntax}}%
  \\
  Structures
   & $\vP,\vQ,\vR,\vS$
   & \AgdaRef{MAV.Structure}+Structure+
  \\
  Duality
   & $\vDual\vP$
   & \AgdaRef{MAV.Structure}+`¬_+[\`{}¬\_]
  \\
  Symmetric MAV
   &
   &
  \\
  \quad Equality
   & $\vEquiv$
   & \AgdaRef{MAV.Symmetric}+_≃_+
  \\
  \quad Inference
   & $\vInferFrom$
   & \AgdaRef{MAV.Symmetric}+_⟶_+
  \\
  \quad Derivation
   & $\vInferFrom*$
   & \AgdaRef{MAV.Symmetric}+_⟶⋆_+
  \\
  \quad Invertible derivation
   & $\vInferFromTo*$
   & \AgdaRef{MAV.Symmetric}+_⟷⋆_+
  \\
  Normal MAV
   & \Cref{defn:normal}
   &
  \\
  \quad Equality
   & $\vEquiv$
   & \AgdaRef{MAV.Base}+_≃_+
  \\
  \quad Inference
   & $\vInferFrom$
   & \AgdaRef{MAV.Base}+_⟶_+
  \\
  \quad Derivation
   & $\vInferFrom*$
   & \AgdaRef{MAV.Base}+_⟶⋆_+
  \\
  \quad Invertible derivation
   & $\vInferFromTo*$
   & \AgdaRef{MAV.Base}+_⟷⋆_+
  \\[2ex]
  \longtableheader{\Cref{sec:mav-semantics-preliminaries}}%
  \\
  Pomonoids
   & \Cref{defn:pomonoid}
   & \AgdaRef*{Algebra.Ordered.Structures}+IsPomonoid+
  \\
  \quad if Commutative
   & \Cref{defn:pomonoid}
   & \AgdaRef*{Algebra.Ordered.Structures}+IsCommutativePomonoid+
  \\
  \quad if Residuated
   & \Cref{defn:residual}
   & \AgdaRef*{Algebra.Ordered.Structures}+IsResiduatedCommutativePomonoid+
  \\
  \quad if Duoidal
   & \Cref{defn:duoidal}
   & \AgdaRef*{Algebra.Ordered.Structures}+IsDuoidal+
  \\
  $*$-autonomous partial order
   & \Cref{defn:star-autonomous}
   & \AgdaRef*{Algebra.Ordered.Structures}+IsStarAutonomous+
  \\
  \longtablemodule{Algebra.Orderd.Structures}%
  \\[2ex]
  \longtableheader{\Cref{sec:mav-algebras}}%
  \\
  MAV-algebra
   & \Cref{defn:mav-algebra}
   & \AgdaRef{MAV.Model}+Model+
  \\
  Interpretation
   & \Cref{defn:mav-interpretation}
   & \AgdaRef{MAV.Interpretation}+⟦_⟧+
  \\
  Soundness
   & \Cref{thm:soundness}
   &
  \\
  \quad for Equality
   &
   & \AgdaRef{MAV.Interpretation}+⟦_⟧eq+
  \\
  \quad for Inference
   &
   & \AgdaRef{MAV.Interpretation}+⟦_⟧step+
  \\
  \quad for Derivation
   &
   & \AgdaRef{MAV.Interpretation}+⟦_⟧steps+
  \\[2ex]
  \longtableheader{\Cref{sec:mav-frames}}%
  \\
  MAV-frame
   & \Cref{defn:mav-frame}
   & \AgdaRef{MAV.Frame}+Frame+
  \\
  Normal MAV-frame
   & \Cref{prop:nmav-frame}
   & \AgdaRef{MAV.Base}+frame+
  \\[2ex]
  \longtableheader{\Cref{sec:lower-sets}}%
  \\
  Lower sets
   & \Cref{defn:lower-set}
   & \AgdaRef*{Algebra.Ordered.Construction.LowerSet}+LowerSet+
  \\
  Day pomonoid
   & \Cref{prop:day-construction}
   & \texttt{module} \AgdaRef*{Algebra.Ordered.Construction.LowerSet}+Day+
  \\
  \quad if Commutative
   & \Cref{prop:day-construction}
   & \texttt{module} \AgdaRef*{Algebra.Ordered.Construction.LowerSet}+DayCommutative+
  \\
  \quad if Duoidal
   & \Cref{prop:lower-set-duoidal}
   & \texttt{module} \AgdaRef*{Algebra.Ordered.Construction.LowerSet}+DayDuoidal+
  \\
  \longtablemodule{Algebra.Ordered.Construction.LowerSet}%
  \\[2ex]
  \longtableheader{\Cref{sec:closed-lower-sets}}%
  \\
  +-closed lower sets
   & \Cref{defn:closed-lower-set}
   & \AgdaRef*{Algebra.Ordered.Construction.Ideal}+Ideal+
  \\
  +-closed Day pomonoid
   &
   &
  \\
  \quad if Commutative
   & \Cref{prop:closed-monoid-distrib}
   & \texttt{module} \AgdaRef*{Algebra.Ordered.Construction.Ideal}+DayCommutative+
  \\
  \quad if Duoidal
   & \Cref{prop:closed-monoid-duoidal}
   & \texttt{module} \AgdaRef*{Algebra.Ordered.Construction.Ideal}+DayDuoidal+
  \\
  \longtablemodule{Algebra.Ordered.Construction.Ideal}%
  \\[2ex]
  \longtableheader{\Cref{sec:chu}}%
  \\
  Chu construction
   & \Cref{defn:chu}
   & \AgdaRef*{Algebra.Ordered.Construction.Chu}+Construction.Chu+
  \\
  \quad for duoidal pomonoids
   & \Cref{prop:chu-monoid-duoidal}
   & \texttt{module} \AgdaRef*{Algebra.Ordered.Construction.Chu}+Construction.SelfDual+
  \\
  \longtablemodule{Algebra.Ordered.Construction.Chu}%
  \\[2ex]
  \longtableheader{\Cref{sec:algebra-from-frame}}%
  \\
  MAV-algebras from MAV-frames
   & \Cref{thm:algebra-from-frame}
   & \AgdaRef{MAV.Frame}+FrameModel.model+
  \\[2ex]
  \longtableheader{\Cref{sec:mav-cut-elimination}}%
  \\
  Main lemma
   & \Cref{prop:embedding-sem}
   & \AgdaRef{MAV.CutElim}+main-lemma+
  \\
  Generalised cut-elimination
   & \Cref{thm:cut-elim}
   & \AgdaRef{MAV.CutElim}+cut-elim+
\end{longtable}

\end{document}